\documentclass[twocolumn]{autart}
\usepackage{graphicx}

\makeatletter
\let\c@author\relax
\makeatother

\usepackage{amsmath}
\usepackage{amssymb}
\usepackage{mathtools}
\usepackage{amsfonts}

\usepackage{physics}
\usepackage{bm}
\usepackage{xcolor}
\usepackage{wrapfig}
\usepackage[hidelinks]{hyperref}

\usepackage[backend=biber,style=authoryear,sorting=nyt,dashed=false, url=false, isbn=false, maxbibnames=5,maxcitenames=5, natbib=true, doi=false]{biblatex}
\DeclareDelimFormat{finalnamedelim}{\ifnumgreater{\value{liststop}}{2}{,\space} \space\&\space}

\begin{document}

\begin{frontmatter}

\title{Estimation Problems and the Modulating Function Method:\\ The Algebra of Modulating Functions}

\thanks[footnoteinfo]{The material in this paper was not presented at any conferences.}

\author[SDU]{Davi G. Accioli}\ead{davi@sdu.dk},
\author[SDU]{Jerome Jouffroy}\ead{jerome@sdu.dk}

\address[SDU]{SDU Mechatronics, University of Southern Denmark (SDU), Alsion 2, 6400, S{\o}nderborg, Denmark}

\begin{keyword}                         
abstract algebra; modulating function method; parameter estimation; system identification.            
\end{keyword}

\begin{abstract}
The modulating function method is an algebraic estimation framework that has been used for state and parameter estimation, as well as fault detection, of linear and some nonlinear systems. At the core of the method is the modulating function: a function that evaluates to 0 at the left or right boundaries up to a certain order of derivatives, being classified as left, right, or total modulating functions. Despite being an algebraic framework, there is no paper that analyzes the algebraic properties of the three types of modulating functions and the consequences of these properties. In light of this gap, this paper discusses the algebraic properties of modulating functions and, after formalizing their closedness and group properties, simple algorithms to construct new modulating functions are proposed, discussed, and illustrated. Moreover, the fact that total modulating functions form an algebra and a vector space is exploited to construct orthonormal modulating functions, which are then used to significantly improve the well-known matrix conditioning issue of the modulating function method, illustrated with the parameter estimation of a boat's nonlinear roll dynamics.
\end{abstract}

\end{frontmatter}

\section{Introduction}\label{sec: introduction}

Initially introduced in \cite{Shin54}, the modulating function method (MFM) is an algebraic estimation method that was first introduced for continuous-time parameter estimation. Since then, it has been extended to various lumped system estimation problems, such as state estimation \citep{LiuKPG14}, as well as simultaneous state and parameter estimation \citep{JouffroyR15} and fault detection \citep{Fischer21}, being also relevant for the PDE and fractional-order sections of the control community \citep{Fischer_flatness_based_PDE_2020, Folke_PDE_Advection_Diffusion_Reaction, Liu_Fractional_Order, Kirati_Fractional_PDE18}.

In addition, this framework generates well-known observer structures, such as the Luenberger observer \citep{D_Acc_MFM_to_Obs}, was recently extended to discrete-time systems, generating the Kalman filter and least-squares parameter estimation \citep{Acc_DT_MFM}, and was further extended to control problems in \cite{Acc_MF_Control}, generating well-known state and output feedback laws as well as a novel fixed-time control law for LTV MIMO systems.

The method's essence, however, lies within the modulating function (MF): a user-defined function that, together with its derivatives, evaluates to zero on the left or right boundaries and, based on which boundaries are zero, these functions can be classified as left, right, or total modulating functions \citep{JouffroyR15}. Precisely due to the selected boundary conditions, it is then possible to simplify the resulting equation using integration by parts: a process that can be elegantly formulated by defining a modulation operator, such as in \cite{Ungarala_MFM_Bioprocesses}. In addition, the MF plays a crucial role in the estimation process by defining the filter characteristics \citep{MaletinskySplineFrequencyDomain}.

Although the general MF selection problem is still an open research question, there are two main approaches in the literature to this problem: select parameters of a know family of MFs \citep{Shin54, Acc_Hyperbolic_MF}; or obtaining the MFs as the solution to an auxiliary problem, as first introduced in \citep{SchmidR11}. Despite the auxiliary problem approach providing a modulating function that is optimal in some sense, direct MF selection is still the most common method in the literature, e.g. \cite{JouffroyR15, Noack_Polynomial_Estimation, IonesiRJ19}, possibly due to the simplicity of implementation. In particular, the approach proposed in \cite{SchmidR11} requires solving an optimal-control problem at every time-step to obtain the MFs, while direct selection allows one to directly sample and store the modulating functions, leading to a simple FIR structure.

Throughout the years, several MF families were introduced \citep{Shin54, LoebCahen65, MaletinskySplineFrequencyDomain, Pearson_Complex_Fourier_MF, Patra_Hartley_MF, Acc_Hyperbolic_MF}, with many of them displaying a similar structure: a product of a term that satisfies the left boundary conditions with a term that satisfies the right boundary conditions, having a possible weight function associated with the result. Even though this notion is not new in the modulating function community, there is currently no study formally discussing the algebraic structure of left, right, and total modulating functions as well as their construction. This gap is particularly significant because it is, e.g., necessary to have the MFs linearly independent from each other in parameter estimation problems \citep{Co_Ungarala_Batch_Recursive_Parameter_Estimation_MFM}. Furthermore, numerical issues may arise if these functions are orthogonal to the measured signal (see, e.g., \cite{Co_Ungarala_Batch_Recursive_Parameter_Estimation_MFM}; \cite{IonesiRJ19}; and \cite{Li_Pin_Kernel_Based_Simulatenous_Parameter_State_Estimation}). Thus, by ensuring at least one of these conditions by construction, the method's applicability for higher-order systems can be improved without increasing computational costs during runtime.

To fill this knowledge gap, this paper formally discusses the algebraic properties of modulating functions. In doing so, a simple algorithm to construct a modulating function from \textit{any} sufficiently smooth function is obtained, and the discussed algebraic properties are exploited to significantly decrease numerical errors associated with matrix inversion. In particular, the main contributions of this paper are:
\begin{itemize}
    \item To formalize the algebraic properties of the sets of modulating functions;
    \item To discuss, formalize, and illustrate the use of multiplication and addition to construct new modulating functions from existing ones;
    \item To propose a simple algorithm to construct \textit{any} type of modulating function of \textit{any} order from \textit{any} sufficiently smooth function, illustrating its use, and discussing its implications;
    \item To introduce a new technique to improve the matrix conditioning for parameter estimation problems, a well-known problem in the MFM community;
    \item To illustrate the applicability of these concepts with the parameter estimation of a boat's roll dynamics.
\end{itemize} 

After this introduction, Section \ref{sec: preliminaries} familiarizes the reader with the MFM and its applications, in particular for parameter estimation, together with a brief overview of some modulating function families proposed over the years. Then, Section \ref{sec: building blocks} discusses the closedness and group properties of modulating functions, culminating in an algorithm to generate MFs from any sufficiently smooth function. Next, the consequences of the algebraic properties on matrix conditioning are discussed, and the concept of orthonormal TMFs is proposed and discussed in Section \ref{sec: matrix conditioning}, having all of the introduced concepts utilized in an application example in Section \ref{sec: application example}. Lastly, some concluding remarks are made in Section \ref{sec: conclusion}.
\section{Preliminaries}\label{sec: preliminaries}

In this section, modulating functions, their main properties, and their application to parameter estimation problems are briefly introduced, together with some known families of modulating functions. Since Hermite functions \citep{Takaya_Hermite_MF} and the Poisson Moment \citep{Saha_Poisson_Moment} \textit{tend} to $0$ at the boundary but are \textit{not} $0$, they do not satisfy the upcoming definitions for finite time and are not further discussed.

\subsection{The Modulating Function Method}

The first step in discussing the modulating function method is to introduce the concept of modulating functions. 

\begin{defn}\label{def: modulating function}
    Given $q \in \mathbb{N}_{>0}$ and a $\mathcal{C}^q([0, T])$ smooth function $\varphi: \mathbb{R}[0, T] \mapsto \mathbb{R}$, $\varphi$ is called a modulating function of order $q$ if it satisfies \citep{JouffroyR15}
    \begin{equation}
        \varphi^{(i)}(0) \cdot \varphi^{(i)}(T)=0
    \end{equation}
    for $i\in \mathbb{N}[0, q-1]$ and $T \in \mathbb{R}_{>0}$. Moreover,
    \begin{itemize}
        \item if $\varphi^{(i)}(0)=\varphi^{(i)}(T)=0~\forall i$, then it is called a total modulating function (TMF), denoted $\varphi_{\mathrm{T}}$;
        \item if $\varphi^{(i)}(0) = 0~\forall i$ and $ \exists i \text{ s.t. }\varphi^{(i)}(T)\neq 0$, then it is called a left modulating function (LMF), denoted $\varphi_{\mathrm{L}}$;
        \item if $\varphi^{(i)}(T)=0 ~\forall i$ and $\exists i \text{ s.t. }\varphi^{(i)}(0)\neq 0$, then it is called a right modulating function (RMF), denoted $\varphi_{\mathrm{R}}$.
    \end{itemize}
\end{defn}

Modulating functions are often employed within integral operators, and by exploiting concepts from linear operator theory, the modulation operator can be defined.

\begin{defn} \label{def: modulation operator}
    Let $\varphi: \mathbb{R}[0, T] \mapsto \mathbb{R}$ be a modulating function of sufficient order $q\geq i \in \mathbb{N}$. Then, a modulation operator is given by \citep{Ungarala_MFM_Bioprocesses}
    \begin{equation}\label{eq: modulation operator def}
        \mathcal{M}^i[\mathrm{y}(t)](t_0, t_1):= \int_{t_0}^{t_1} (-1)^i \varphi^{(i)}(\tau-t_0) \mathrm{y}(\tau) \dd \tau \text{,}
    \end{equation}
    for user-defined $t_0, t_1 \in \mathbb{R}$, $t_1>t_0$, and $T:=t_1-t_0$.
\end{defn}

\begin{rem}
    The dependency of the modulation operator on $t_0$ and $t_1$ is henceforth omitted to avoid clutter, along with the introduction of the notation $\mathcal{M}[\mathrm{y}]:={\mathcal{M}}^0[\mathrm{y}]$. Moreover, modulation operators that utilize left, right, and total modulating functions are respectively denoted as $\mathcal{M}_{\mathrm{L}}$, $\mathcal{M}_{\mathrm{R}}$, and $\mathcal{M}_{\mathrm{T}}$, which may also be numbered, e.g. $\prescript{}{i}{\mathcal{M}}$, if different MFs, e.g., $\prescript{}{i}{\varphi}$ or $\prescript{}{i}{\psi}$ are used.
\end{rem}

The main application of the modulation operator within the MFM framework is to transfer derivatives of measured signals to the modulating function. In doing so, one effectively avoids differentiation of noisy signals.

\begin{lem}\label{lemma: transfer of derivatives}
    Let $\varphi \in \mathcal{C}^i([0, T])$. Then, applying modulation operator \eqref{eq: modulation operator def} to the derivative of a signal transfers the derivative to the modulating function, i.e.
    \begin{equation}\label{eq: transfer of derivatives}
        {\mathcal{M}} [\mathrm{y}^{(i)}]={\mathcal{M}}^i [\mathrm{y}]+P \text{,}
    \end{equation}
    where
    \begin{equation}\label{eq: transfer of derivatives boundary conditions}
        P:=\sum_{j=0}^{i-1} \begin{array}{l}
             (-1)^j \varphi^{(j)}(T)\mathrm{y}^{(i-1-j)}(t_1) \\
             - (-1)^j\varphi^{(j)}(0)\mathrm{y}^{(i-1-j)}(t_0)
        \end{array} 
    \end{equation}
\end{lem}
\begin{pf}
    Apply modulation operator \eqref{eq: modulation operator def} to the $i$-th derivative of a signal $\mathrm{y}(t)$, denoted $\mathrm{y}^{(i)}(t)$. Then, expand the definition of the modulation operator and successively perform integration by parts, leading to \eqref{eq: transfer of derivatives}-\eqref{eq: transfer of derivatives boundary conditions}. \hfill $\blacksquare$
\end{pf}

By using the fact that modulation operator \eqref{eq: modulation operator def} is a linear operator, problems involving linear systems can be elegantly reformulated. In particular, consider a SISO LTI system in input-output form given as
\begin{equation}
    \mathrm{y}^{(n)} + \sum_{i=0}^{n-1} a_{i}\mathrm{y}^{(i)}=\sum_{i=0}^{n-1} b_{i}\mathrm{u}^{(i)}\text{,}
    \label{eq: LTI siso input-output form}
\end{equation}
where $a_i, ~b_i \in \mathbb{R}$. \\
For parameter estimation problems, one can easily use TMFs to form a system of linear equations and obtain the parameter estimate in fixed time \citep{Shin54, JouffroyR15}. If it is instead desired to estimate states or output derivatives, left and right MFs can be employed, also leading to an estimate in fixed time \citep{LiuKPG14, IonesiRJ19, Noack_Lagrangian_MF_Estimation}. In this paper, the focus is on parameter estimation, which is a more challenging problem than state or derivative estimation \citep{KorderNR22}.

\begin{thm}\label{theo: parameter estimation}
    Given a system described by \eqref{eq: LTI siso input-output form}, modulation operator \eqref{eq: modulation operator def}, and linearly independent modulating functions not all simultaneously orthogonal to the measured signals $\mathrm{y}(t), \mathrm{u}(t) \in \mathbb{R}$, the system parameters can be estimated as \citep{Co_Ungarala_Batch_Recursive_Parameter_Estimation_MFM}
    \begin{equation}\label{eq: parameter estimation}
        \hat{\bm{\theta}} =\mathbf{W}^{-1} \mathbf{z}\text{,}
    \end{equation}
    where
    \begin{gather}
        \hat{\bm{\theta}}:=\left[ \hat{a}_0, \hat{a}_1, ~\ldots~, \hat{a}_{n-1}, \hat{b}_0, ~\ldots~, \hat{b}_{n-1}\right]^\top \text{,} \\
        \mathbf{z}:=\left[\prescript{}{1}{\mathcal{M}_\mathrm{T}^{n}[\mathrm{y}]},  ~\ldots~, \prescript{}{2n}{\mathcal{M}_\mathrm{T}^{n}[\mathrm{y}]}  \right]^\top \text{,} \label{eq: definition z parameter estimation} \\
        \label{eq: definition W parameter estimation}
        \mathbf{W}:= \begin{bmatrix} -\mathbf{W}_y & \mathbf{W}_u\end{bmatrix}\text{,}
    \end{gather}
    for
    \begin{equation}\label{eq: definition Wy parameter estimation}
        \mathbf{W}_y:=\begin{bmatrix} \prescript{}{1}{\mathcal{M}_\mathrm{T}[\mathrm{y}]} & \cdots & \prescript{}{1}{\mathcal{M}_\mathrm{T}^{n-1}[\mathrm{y}]} \\
        \prescript{}{2}{\mathcal{M}_\mathrm{T}[\mathrm{y}]} & \cdots & \prescript{}{2}{\mathcal{M}_\mathrm{T}^{n-1}[\mathrm{y}]} \\
        \vdots & \ddots & \vdots \\
        \prescript{}{2n}{\mathcal{M}_\mathrm{T}[\mathrm{y}]} & \cdots & \prescript{}{2n}{\mathcal{M}_\mathrm{T}^{n-1}[\mathrm{y}]} \end{bmatrix},
    \end{equation}
    \begin{equation}\label{eq: definition Wu parameter estimation}
        \mathbf{W}_u:=\begin{bmatrix} \prescript{}{1}{\mathcal{M}_\mathrm{T}[\mathrm{u}]} & \cdots & \prescript{}{1}{\mathcal{M}_\mathrm{T}^{n-1}[\mathrm{u}]} \\
        \prescript{}{2}{\mathcal{M}_\mathrm{T}[\mathrm{u}]} & \cdots & \prescript{}{2}{\mathcal{M}_\mathrm{T}^{n-1}[\mathrm{u}]} \\
        \vdots & \ddots & \vdots \\
        \prescript{}{2n}{\mathcal{M}_\mathrm{T}[\mathrm{u}]} & \cdots & \prescript{}{2n}{\mathcal{M}_\mathrm{T}^{n-1}[\mathrm{u}]} \end{bmatrix}\text{.}
    \end{equation}
\end{thm}
\begin{pf}
    Start by applying a total modulation operator \eqref{eq: modulation operator def} to \eqref{eq: LTI siso input-output form}, using Lemma \ref{lemma: transfer of derivatives}, and reorganizing to obtain
    \begin{equation}
        \mathcal{M}^n_\mathrm{T}[\mathrm{y}(t)]= \sum_{i=0}^{n-1} -a_i\mathcal{M}^i_\mathrm{T}[\mathrm{y}] +b_i\mathcal{M}^i_\mathrm{T}[\mathrm{u}]\text{,}
    \end{equation}
    where the boundary condition terms from \eqref{eq: transfer of derivatives} are all $0$ due to TMFs being used.\\
    Next, repeat the process for a total of $2n$ equations and rewrite the summation as a matrix multiplication to obtain $\mathbf{z}=\mathbf{W}\bm{\theta}$, where $\mathbf{W}$ is given by \eqref{eq: definition W parameter estimation}-\eqref{eq: definition Wu parameter estimation}, $\mathbf{z}$ is given by \eqref{eq: definition z parameter estimation}, and $\bm{\theta}:=\left[a_0, a_1, \ldots, a_{n-1}, b_0, b_1, \ldots, b_{n-1} \right]^\top$. \\
    Lastly, if the MFs are linearly independent and not all orthogonal to the measured signals, then the resulting column vectors of modulated signals are linearly independent and $\det \mathbf{W}\neq 0$ \citep{IonesiRJ19}, leading to \eqref{eq: parameter estimation} as the unique solution. $\hfill \blacksquare$
\end{pf}

This method converges in a fixed-time $T$ \citep{KorderNR22, Byrski03} and is robust to noise \citep{Noack_Polynomial_Estimation}, but it is crucial that the MFs are linearly independent, otherwise $\mathbf{W}$ in \eqref{eq: parameter estimation} becomes non-invertible, and the parameters cannot be estimated.
One approach to this problem is to obtain the modulating function itself as the solution to an auxiliary system, as previously discussed (see also \cite{Noack_History_Auxiliary_Systems_MFM}), but many contributions simply select a function from a previously known repertoire of modulating functions. To increase this repertoire, many families of MFs have been proposed, each of them with their particular properties.

\begin{rem}
    It is important to note that the following discussion is not an exhaustive list, and the modulating functions explicitly given here were selected mainly based on their algebraic structure and formulation.
\end{rem}

\subsection{Brief Overview of Known Modulating Functions}

The first family of modulating functions to be proposed was the sine-type TMF family, given by
\begin{equation}\label{eq: sine TMF family}
    \varphi(\tau)=\sin^q\left( \frac{q \pi}{T} \tau \right)\text{,}
\end{equation}
where $q \in \mathbb{N}_{>0}$ is the order of the modulating function and $T \in \mathbb{R}_{>0}$ is the integration window. This MF family was introduced together with the notion of the modulating function method in \cite{Shin54}.\\
About 10 years later, the polynomial family of modulating functions was introduced in \cite{LoebCahen65}, being of the form
\begin{equation}\label{eq: polynomial MF family}
    \varphi(\tau)=F(\tau)\tau^{q_1}(\tau-T)^{q_2}\text{,}
\end{equation}
where the tuning parameters $q_1, q_2 \in \mathbb{N}_{>0}$ dictate the order of the modulating function and $F: \mathbb{R} \mapsto \mathbb{R}$ is a user-defined sufficiently smooth weight function. Unlike the sine family, polynomial MFs can generate left, right, and total modulating functions, and can be used for different types of estimation problems. For this reason, they are ubiquitous in the literature, see, e.g., \cite{JouffroyR15, IonesiRJ19, Noack_Polynomial_Estimation}.\\
Then, after the spline family of modulations introduced by \cite{MaletinskySplineFrequencyDomain}, the complex Fourier family of TMFs was introduced in \cite{Pearson_Complex_Fourier_MF}, being given as
\begin{equation}\label{eq: complex-Fourier TMF}
    \varphi(\tau)=\frac{1}{T}e^{- i m\omega_0\tau}\left(e^{-i\omega \tau} -1\right)^q\text{,}
\end{equation}
based on the Fourier transform and frequency domain analysis, where $i:=\sqrt{-1}$ is the imaginary unit, $m \in \mathbb{Z}$, and $\omega_0=2\pi/T$. Soon after, the Hartley family of total modulating functions was introduced in \cite{Patra_Hartley_MF}, also focusing on the frequency domain characteristics of the resulting filter.\\
More recently, two more LMFs were introduced in \cite{Pin_exponential_MF}, namely the exponential left modulating function 
\begin{equation}\label{eq: Pin and Parisini exp RMF}
    \varphi(\tau)=(1-e^{-c\tau})^{q}
\end{equation}
and a time-monomial modulating function.\\
Lastly, a total of 16 families of hyperbolic modulating functions were introduced by \cite{Acc_Hyperbolic_MF}, with the 4 main families being the hyperbolic sine family
\begin{equation}\label{eq: sinh family of modulations}
    \varphi(\tau)= F(\tau)\sinh^{q_1}(c_1 \tau) \sinh^{q_2}(c_2(T-\tau))\text{;}
\end{equation}
the hyperbolic cosine family
\begin{equation}\label{eq: cosh family of modulations}
    \varphi(\tau)=F(\tau)(\cosh(c_1 \tau)-1)^{q_1} (\cosh(c_2(T-\tau))-1)^{q_2}\text{;}
\end{equation}
the hyperbolic tangent family
\begin{equation}\label{eq: tanh family of modulations}
    \varphi(\tau)= F(\tau)\tanh^{q_1}(c_1 \tau) \tanh^{q_2}(c_2(T-\tau))\text{;}
\end{equation}
and the hyperbolic secant family
\begin{equation}\label{eq: sech family of modulations}
    \varphi(\tau)=F(\tau)(\sech(c_1 \tau)-1)^{q_1} (\sech(c_2(T-\tau))-1)^{q_2}\text{,}
\end{equation}
all of which generate left, right, and total modulating functions with tuning parameters $c_1, c_2 \in \mathbb{R}_{\neq 0}$.

Another topic discussed in \cite{Acc_Hyperbolic_MF} is that it is possible to use different hyperbolic functions to satisfy the left and right boundary conditions, leading to the other 12 MF families of mixed hyperbolic modulations, where the boundary conditions are satisfied through the product of two different functions that satisfy either left or right boundary condition.
\section{The Algebra of Modulating Functions}
\label{sec: building blocks}

While the idea that, e.g., new total modulating functions can be obtained by adding and multiplying previously known MFs is not new (dating back to \cite{Perdreauville_PhD_Thesis}), the algebraic properties of left, right, and total MFs have not been formally discussed within the modulating function community yet. Interestingly, not only do MFs form a semigroup, but TMFs also form an algebra.

\begin{thm}\label{theo: closedness MFs}
    Let $\Phi_{\mathrm{L}}$, $\Phi_{\mathrm{R}}$, and $\Phi_{\mathrm{T}}$ respectively be the sets of all left, right, and total modulating functions, and $\Phi$ be the set of modulating functions, i.e. $\Phi=\{\Phi_{\mathrm{L}}, \Phi_{\mathrm{R}}, \Phi_{\mathrm{T}} \}$. Then, it follows that
    \begin{itemize}
        \item $\Phi_{\mathrm{L}}$, $\Phi_{\mathrm{R}}$, $\Phi_{\mathrm{T}}$, and $\Phi$ form abelian semigroups under point-wise multiplication;
        \item $\Phi_{\mathrm{L}}$ and $\Phi_{\mathrm{R}}$ are closed under point-wise multiplication with any sufficiently smooth nonzero function $F(\tau)$;
        \item $\Phi_{\mathrm{T}}$ and $\Phi$ are closed under point-wise multiplication with any sufficiently smooth function $F(\tau)$;
        \item $\Phi_{\mathrm{T}}$ forms an abelian group under point-wise addition, a vector space, and an associative and commutative algebra.
    \end{itemize}
\end{thm}

\begin{pf}
    To avoid cluttering, the time dependency of the modulating functions is often omitted. \\
    Let $\prescript{}{1}{\varphi}_{\mathrm{L}}(\tau) \in \mathcal{C}^q([0, T])$ be an LMF of order $q_1 \in \mathbb{N}_{>0}$ and $\prescript{}{2}{\varphi}_{\mathrm{L}}(\tau) \in \mathcal{C}^q([0, T])$ be an LMF of order $q_2 \in \mathbb{N}_{>0}$, for $q=\max(q_1, q_2)$. Then, their point-wise product, which is also smooth, and the derivatives of their product can be evaluated using the product rule
    \begin{equation}
        \left( \prescript{}{1}{\varphi_{\mathrm{L}}}~\prescript{}{2}{\varphi_{\mathrm{L}}} \right)^{(k)}=\sum_{i=0}^{k} \begin{pmatrix} k \\ i \end{pmatrix} \prescript{}{1}{\varphi_{\mathrm{L}}}^{(k-i)} \prescript{}{2}{\varphi_{\mathrm{L}}}^{(i)}\text{,}
    \end{equation}
    having $k \in \mathbb{N}[0, q_1+q_2-1]$. Evaluating the product at the left boundary, it satisfies the left boundary condition
    \begin{equation}
        \left( \prescript{}{1}{\varphi_{\mathrm{L}}}~\prescript{}{2}{\varphi_{\mathrm{L}}} \right)^{(k)}(0)=\sum_{i=0}^{k} \begin{pmatrix} k \\ i \end{pmatrix} \prescript{}{1}{\varphi_{\mathrm{L}}}^{(k-i)}(0) \prescript{}{2}{\varphi_{\mathrm{L}}}^{(i)}(0) =0
    \end{equation}
    by definition, as every combination of derivative $\prescript{}{1}{\varphi_{\mathrm{L}}}^{(k-i)}(0)$ of order $k-i\geq q_1$ is paired with a derivative $\prescript{}{2}{\varphi_{\mathrm{L}}}^{(i)}(0)$ where $i<q_2$.
    Next, it is necessary to prove that the right boundary is non-zero for some $k$. To do so, evaluate the product at $T$, i.e.
    \begin{equation}
        \left( \prescript{}{1}{\varphi_{\mathrm{L}}}~\prescript{}{1}{\varphi_{\mathrm{L}}} \right)^{(k)}(T)=\sum_{i=0}^{k} \begin{pmatrix} k \\ i \end{pmatrix} \prescript{}{1}{\varphi_{\mathrm{L}}}^{(k-i)}(T) \prescript{}{2}{\varphi_{\mathrm{L}}}^{(i)}(T)\text{,}
    \end{equation}
    and note that all index combinations are uniquely present in the derivatives from $0$ to $q_1+q_2-1$. Since at least one derivative of $\prescript{}{1}{\varphi_{\mathrm{L}}}$ and $\prescript{}{2}{\varphi_{\mathrm{L}}}$ is non-zero at the right boundary, at least one non-zero product is obtained at the right boundary. Thus, the product of two LMFs is an LMF, i.e. $\prescript{}{1}{\varphi_{\mathrm{L}}} \prescript{}{2}{\varphi_{\mathrm{L}}}=\prescript{}{3}{\varphi_{\mathrm{L}}}$, and the set $\Phi_{\mathrm{L}}$ is closed under the associative operation of point-wise multiplication. Thus, it is an abelian semigroup under point-wise multiplication. The same argument applies to the set $\Phi_{\mathrm{R}}$ and for multiplication with sufficiently smooth nonzero functions $F(\tau)$. For TMFs, it suffices to show that, by definition,
    \begin{equation}
        \prescript{}{1}{\varphi_{\mathrm{T}}}(0) \prescript{}{1}{\varphi_{\mathrm{T}}}(0)=\prescript{}{1}{\varphi_{\mathrm{T}}}(T) \prescript{}{1}{\varphi_{\mathrm{T}}}(T)=0 \cdot 0 =0 \text{.}
    \end{equation}
    Thus, the product of any two TMFs is a TMF of order at least $1$, but higher order terms can be discussed with the same arguments as for LMFs and RMFs, with the exception that now the function $F(\tau)$ does not have to be nonzero. In particular, the zero function trivially satisfies the boundary conditions of an infinite-order TMF.\\    
    Next, consider the product of an LMF with an RMF, evaluate it at the boundaries, and note that, again,
    \begin{gather}
        \varphi_{\mathrm{L}}(0)  \varphi_{\mathrm{R}}(0) = 0 \cdot  \varphi_{\mathrm{R}}(0)=0 \\
        \varphi_{\mathrm{L}}(T)  \varphi_{\mathrm{R}}(T) = \varphi_{\mathrm{L}}(T) \cdot 0=0\text{.}
    \end{gather}
    Thus, the product of an LMF and an RMF is, at least, a TMF of order $1$. Then, take the product of a TMF with an LMF or RMF and evaluate it at the boundaries, i.e.
    \begin{gather}
        \varphi_{\mathrm{T}}(0)  \varphi_{\mathrm{L}}(0) = 0 \cdot 0=0 \\
        \varphi_{\mathrm{T}}(T)  \varphi_{\mathrm{L}}(T) = 0 \cdot \varphi_{\mathrm{L}}(T)=0 \\
        \varphi_{\mathrm{T}}(0)  \varphi_{\mathrm{R}}(0) = 0 \cdot \varphi_{\mathrm{R}}(0)=0 \\
        \varphi_{\mathrm{T}}(T)  \varphi_{\mathrm{R}}(T) = 0 \cdot 0=0\text{.}
    \end{gather}
    Next, use the fact that point-wise multiplication is associative as well as commutative and that the product of a sufficiently smooth function $F(\tau)$ with any modulating function cannot remove existing zero boundary conditions, by the same arguments as above, to conclude that the set of modulating functions $\Phi$ is closed under point-wise multiplication, forming an abelian, i.e. commutative, semigroup.\\
    Lastly, note that the set of TMFs $\Phi_\mathrm{T}$ includes the trivial solution $\varphi_{\mathrm{T}}(\tau)=0$ and is closed under point-wise addition since
    \begin{gather}
        \prescript{}{1}{\varphi_{\mathrm{T}}}(0)+\prescript{}{2}{\varphi_{\mathrm{T}}}(0)=0+0=0\\
        \prescript{}{1}{\varphi_{\mathrm{T}}}(T)+\prescript{}{2}{\varphi_{\mathrm{T}}}(T)=0+0=0 \text{.}
    \end{gather}
    Moreover, given any TMF with an arbitrary sufficiently smooth weight function $F_1(\tau)$, the additive inverse is uniquely obtained by selecting $F_2(\tau)=-F_1(\tau)$. Since $\Phi_\mathrm{T}$ has an additive identity, is closed under the associative and commutative operator of point-wise addition, and every TMF has an additive inverse, $\Phi_\mathrm{T}$ forms a commutative, i.e. abelian, group under point-wise addition. Lastly, since $\Phi_\mathrm{T}$ forms an abelian semigroup under point-wise multiplication and also forms an abelian group under addition, it suffices to note that point-wise multiplication distributes under the associative addition operation to conclude that $\Phi_\mathrm{T}$ also forms a vector space and an associative and commutative algebra. \hfill $\blacksquare$
\end{pf}

\begin{rem}
    Note that it is essential that $0=\varphi(\tau)\in \Phi_{\mathrm{T}}\subset \Phi$, otherwise $\Phi_{\mathrm{T}}$ does not have an additive identity and cannot form a group under addition.
\end{rem}

\begin{prop}\label{proposition: sum of LMFs and RMFs}
    Given any two modulating functions, the following properties are satisfied:
    \begin{itemize}
        \item the sum of LMFs is a modulating function;
        \item the sum of RMFs is a modulating function;
        \item the sum of a TMF and an LMF is an MF;
        \item the sum of a TMF and an RMF is an MF.
    \end{itemize}
\end{prop}
\begin{pf}
    Given any two sufficiently smooth LMFs $\prescript{}{1}{\varphi_{\mathrm{L}}}$ and $\prescript{}{2}{\varphi_{\mathrm{L}}}$, the derivatives of their sum is given as
    \begin{equation}
        \left( \prescript{}{1}{\varphi_{\mathrm{L}}}+\prescript{}{2}{\varphi_{\mathrm{L}}}\right)^{(i)} =\prescript{}{1}{\varphi_{\mathrm{L}}}^{(i)}+\prescript{}{2}{\varphi_{\mathrm{L}}}^{(i)}\text{,}
    \end{equation}
    where $i \in \mathbb{N}_{>0}$. By assuming $q_1\leq q_2$ without any loss of generality, the left boundary trivially evaluates to $0$ up to order $q_1-1$, but the right boundary is $0$ iff $\prescript{}{1}{\varphi_{\mathrm{L}}}^{(i)}(T)=-\prescript{}{2}{\varphi_{\mathrm{L}}}^{(i)}(T)$, and otherwise is non-zero. Thus, the right boundary can be $0$ or non-zero up to derivative $q_1-1$, resulting in a TMF or an LMF, respectively. If the right boundary is $0$ up to the $q\geq q_1$-th derivative, then it is an RMF since the left boundary condition of $\prescript{}{1}{\varphi_{\mathrm{L}}}$, i.e. $\prescript{}{1}{\varphi_{\mathrm{L}}}^{(q_1)}(0)$, must be non-zero, given that it is an LMF of order $q_1$. Using the same arguments, the sum of two RMFs can be an LMF, RMF, or TMF, and the sum of a TMF with an LMF or RMF is also an MF. \hfill $\blacksquare$
\end{pf}

\begin{rem}
    Note that, even though the sum of an LMF and an RMF \textbf{may} be a modulating function, the sum of an LMF with an RMF is not \textbf{necessarily} a modulating function. As a simple example, consider an LMF $\prescript{}{1}{\varphi_{\mathrm{L}}}$ of arbitrary order $q_1\geq 1$ that satisfies $\prescript{}{1}{\varphi_{\mathrm{L}}}(T)=1$ and an RMF $\prescript{}{2}{\varphi_{\mathrm{R}}}$ of arbitrary order $q_2\geq 1$ satisfying $\prescript{}{2}{\varphi_{\mathrm{R}}}(0)=1$. Then, regardless of the boundary conditions for higher derivatives, their sum cannot satisfy the $0$-th order derivative boundary condition to be an MF.
\end{rem}

The properties discussed in Proposition \ref{proposition: sum of LMFs and RMFs} and its proof provide an interesting tuning procedure for MFs: given any left or right modulating functions $\prescript{}{1}{\varphi_{\mathrm{L}}} \in \Phi_{\mathrm{L}}$ and $\prescript{}{1}{\varphi_{\mathrm{R}}} \in \Phi_{\mathrm{R}}$ of order $q_1$, they are invariant under addition with any TMF $\varphi_{\mathrm{T}} \in \Phi_{\mathrm{T}}$ of order $q\geq q_1$, i.e. $\prescript{}{1}{\varphi_{\mathrm{L}}}+\varphi_{\mathrm{T}}=\prescript{}{2}{\varphi_{\mathrm{L}}} \in \Phi_{\mathrm{L}}$, $\prescript{}{1}{\varphi_{\mathrm{R}}}+\varphi_{\mathrm{T}}=\prescript{}{2}{\varphi_{\mathrm{R}}} \in \Phi_{\mathrm{R}}$, and $q_2=q_1$. Thus, it is possible to manipulate LMFs and RMFs without changing their order by adding TMFs of the same order.
Similarly, by multiplying MFs together, it is possible to increase the order of the resulting MF, as shown in the proof of Theorem \ref{theo: closedness MFs}. In this sense, it is possible to increase and decrease the order of the modulating function as desired by simple addition and multiplication.

\begin{cor}\label{corollary: increase and decrease order of MFs}
    Let $\prescript{}{1}{q}_{\mathrm{L}}, \prescript{}{1}{q}_{\mathrm{R}} \in \mathbb{N}_{>0}$ be respectively the order of the sufficiently smooth modulating function $\prescript{}{1}{\varphi}(\tau)$ on the left and right boundaries; and $\prescript{}{2}{q}_{\mathrm{L}}, \prescript{}{2}{q}_{\mathrm{R}} \in \mathbb{N}_{>0}$ be the order of the sufficiently smooth modulating function $\prescript{}{2}{\varphi}(\tau)$ on the left and right boundaries, respectively. Then, it follows that
    \begin{itemize}
        \item Combining the two modulating functions using a product, i.e. $\prescript{}{3}{\varphi}(\tau)=\prescript{}{1}{\varphi}(\tau)\prescript{}{2}{\varphi}(\tau)$, increases the order of the left and right boundary conditions to $\prescript{}{3}{q}_{\mathrm{L}}=\prescript{}{1}{q}_{\mathrm{L}}+\prescript{}{2}{q}_{\mathrm{L}}$ and $\prescript{}{3}{q}_{\mathrm{R}}=\prescript{}{1}{q}_{\mathrm{R}}+\prescript{}{2}{q}_{\mathrm{L}}$;
        \item Combining the two modulating functions using addition, i.e. $\prescript{}{3}{\varphi}(\tau)=\prescript{}{1}{\varphi}(\tau)+\prescript{}{2}{\varphi}(\tau)$, changes the order to $\prescript{}{3}{q}_{\mathrm{L}} \geq \min(\prescript{}{1}{q}_{\mathrm{L}}, \prescript{}{2}{q}_{\mathrm{L}})$ and $\prescript{}{3}{q}_{\mathrm{R}} \geq \min(\prescript{}{1}{q}_{\mathrm{R}}, \prescript{}{2}{q}_{\mathrm{R}})$.
    \end{itemize} 
\end{cor}

These properties are applicable to any family of modulating functions and, therefore, allow for new modulating functions to be built using the operations of addition and multiplication in this set of rules. 

\begin{exmp}
    For example, the product of polynomial MFs always results in a polynomial MF, i.e. $\left( \tau^{q_1}(\tau-T)^{q_2} \right) \cdot \left( \tau^{q_3}(\tau-T)^{q_4}\right)=\tau^{q_1+q_3}(\tau-T)^{q_2+q_4}$, and, consequently, is still a modulating function but of a higher order. Similarly, the product of two hyperbolic MFs might not be reducible, as discussed in \cite{Acc_Hyperbolic_MF}, but it is still a modulating function nonetheless.
\end{exmp}

\begin{exmp}
    Another example is that, by adding a sinh and cosh LMF of order $1$ with $F_1=F_2=F$ and $c_1=c_2=c$, one directly obtains an exponential LMF, i.e. $F(\tau) \sinh(c \tau) + F(\tau)(\cosh(c \tau) - 1)=F(\tau)\left( e^{c \tau}-1\right)$.
    By repeating the same procedure for RMFs and combining the two results with repeated multiplication, the family of exponential modulations is obtained through Theorem \ref{theo: closedness MFs}, Proposition \ref{proposition: sum of LMFs and RMFs}, and Corollary \ref{corollary: increase and decrease order of MFs}, being given as
    \begin{equation}\label{eq: exponential family of MF}
        \varphi(\tau)=F(\tau)\left( e^{c_1 \tau}-1\right)^{q_1} \left( e^{c_2 (T- \tau)}-1\right)^{q_2}\text{,}
    \end{equation}
    with the same tuning parameters as \eqref{eq: sinh family of modulations}-\eqref{eq: sech family of modulations}. This result generalizes the left modulating exponential function \eqref{eq: Pin and Parisini exp RMF} introduced in \cite{Pin_exponential_MF}, i.e. $\varphi(\tau)=(1-e^{c\tau})^q$, and can be compared to the complex Fourier TMFs \eqref{eq: complex-Fourier TMF}, showing a notable resemblance. Although $c_1$ and $F(\tau)$ have been assumed to be real-valued, the same boundary conditions are satisfied if they are complex. Therefore, \eqref{eq: complex-Fourier TMF} and \eqref{eq: Pin and Parisini exp RMF} are particular cases of \eqref{eq: exponential family of MF}.
\end{exmp}

A pattern emerges from these MF families: the function minus its evaluation on the left boundary multiplied by the analogous version at the right boundary. Consequently, a simple procedure can be deduced to create a modulating function from any sufficiently smooth function using Theorem \ref{theo: closedness MFs}, Proposition \ref{proposition: sum of LMFs and RMFs}, and Corollary \ref{corollary: increase and decrease order of MFs}.

\begin{thm}\label{theo: build MF}
    Given any sufficiently smooth scalar function $g: \mathbb{R}[0, T] \mapsto \mathbb{R}$, the family of functions
    \begin{equation}\label{eq: algorithm for new modulation}
        \varphi(\tau)=F(\tau)\left( g(\tau)-g(0) \right)^{q_1}\left( g(\tau)-g(T) \right)^{q_2}
    \end{equation}
    is a family of modulating functions, where the sufficiently smooth function $F(\tau) \in \mathbb{R}$, and the coefficients $T \in \mathbb{R}_{>0}$ and $q_1, q_2 \in \mathbb{N}$, are user-defined parameters.
\end{thm}
\begin{pf}
    Start by assuming a sufficiently smooth function $g: \mathbb{R}[0, T] \mapsto \mathbb{R}$ and subtract the left-boundary value to obtain the first-order LMF
    \begin{equation}
        \varphi_{\mathrm{L}}(\tau)=g(\tau)-g(0)\text{.}
    \end{equation}
    Then, multiply $\varphi_{\mathrm{L}}$ with itself for a total of $q_1$ times to increase the order of the MF per Corollary \ref{corollary: increase and decrease order of MFs}. Next, note that a first-order RMF can be generated as
    \begin{equation}
        \varphi_{\mathrm{R}}(\tau)=g(\tau)-g(T)\text{,}
    \end{equation}
    and its order can be increased to $q_2$ through repeated multiplication. Since the product of MFs is also an MF per Theorem \ref{theo: closedness MFs}, a new MF family is obtained as \eqref{eq: algorithm for new modulation}, having also included the sufficiently smooth weight function $F(\tau) \in \mathbb{R}$ in accordance with Theorem \ref{theo: closedness MFs}. \hfill $\blacksquare$
\end{pf}

\begin{rem}
    Although one could obtain the right modulating function as $\varphi_{\mathrm{R}}=g(T-\tau)-g(0)$, the alternative $\varphi_{\mathrm{R}}=g(\tau)-g(T)$ is seldom mentioned, possibly in an attempt to maintain symmetry. However, symmetry is not necessary in any of the results proposed within this document and, therefore, is not a necessary restriction.
\end{rem}

\begin{exmp}
    As an application of Theorem \ref{theo: build MF}, take the natural logarithm function. By shifting the domain to $[1, T+1]$ to satisfy the smoothness requirements\footnote{This could be any compact interval starting at some non-zero $a \in \mathbb{R}_{>0}$ and ending at $T+a$.} and considering a possible scaling factor $c \in \mathbb{R}_{>0}$, the function $\ln (c \tau+1)$ is readily obtained for $\tau \in \mathbb{R}[0, T]$. This function is infinitely differentiable and, despite evaluating to $0$ on the left boundary, all derivatives on the left boundary are non-zero. Consequently, modulating functions of any order $q$ can be obtained by applying Theorem \ref{theo: build MF}, giving rise to the logarithmic family of MFs
    \begin{equation}\label{eq: ln MF family}
        \varphi(\tau)\hspace{-0.07cm}=\hspace{-0.07cm} F(\tau) \bigl(\ln (c_1 \tau+1)\bigr)^{q_1} \hspace{-0.05cm} \bigl(\ln (c_2 \tau+1) - \ln(c_2 T+1)\bigr)^{q_2} \hspace{-0.1cm}\text{.}
    \end{equation}
\end{exmp}

\begin{exmp}
    Since no power series expansion has been used in any of the proofs, it is possible to construct modulating functions that are non-analytic using Theorem \ref{theo: build MF}. Therefore, the space of test functions might be a good candidate for new modulating functions, as they provide infinitely smooth functions with compact support that are not identically zero \citep[p.~10]{Grubb_Distributions}. Thus, consider
    \begin{equation}\label{eq: bump test function}
        h(\tau)=\left\{ \begin{array}{l}
         e^{\frac{-1}{1-\tau^2}}, \text{for } |\tau|<1 \\
         0, \text{for } |\tau| \geq 1
    \end{array} \right.\text{,}
    \end{equation}
    a ubiquitous example of test functions. This function can be easily converted into a family of TMFs of infinite order by simply shifting it one unit to the right, scaling $\tau$ to fit within the time window, and multiplying by a sufficiently smooth function $F(\tau)$, per Theorem \ref{theo: closedness MFs}, to obtain
    \begin{equation}\label{eq: bump TMF family}
        \varphi_{\mathrm{T}}(\tau)=F(\tau)~ h\Bigl(\frac{2\tau}{T}-1\Bigr)\text{.}
    \end{equation}
    Similarly, consider the function $f(\tau)=e^{-\frac{1}{\tau}}$ for $\tau>0$ and $f(\tau)=0$ otherwise, which, by construction, evaluates to zero on the left boundary, also for all derivatives. It is then possible to obtain a smooth transition between two points, similar to \cite{Reger_nonanalytic_functions}, given as
    \begin{equation}
        g(\tau)=\frac{f(\tau)}{f(\tau)+f(1-\tau)}\text{,}
    \end{equation}
    which is a function that starts at $g(0)=0$ and smoothly arrives at $g(1)=1$. Next, scale $\tau$ to the interval $[0, T]$ and consider a sufficiently smooth function $F(\tau)$ that has at least one non-zero derivative on the right boundary to obtain the LMF family $\varphi_{\mathrm{L}}(\tau)= F(\tau) ~ g\Bigl(\frac{\tau}{T}\Bigr)$. Note that the restriction on the function $F(\tau)$, i.e. it must evaluate to a non-zero value at the right boundary for some derivative of finite order, is not present for TMFs given by \eqref{eq: bump TMF family}. Lastly, a family of RMFs can also be generated as $\varphi_{\mathrm{R}}(\tau)=F(\tau)~g\Bigl(1-\frac{\tau}{T}\Bigr)$, where $F(\tau)$ must be non-zero at the left boundary for some finite derivative.
\end{exmp}

\begin{rem}
    Care must be taken when using test functions as TMFs, since they do not form an algebra.
\end{rem}
\section{Improving Matrix Conditioning}
\label{sec: matrix conditioning}

One important use of Theorem \ref{theo: closedness MFs}, Proposition \ref{proposition: sum of LMFs and RMFs}, Corollary \ref{corollary: increase and decrease order of MFs}, and Theorem \ref{theo: build MF} is to improve matrix conditioning. To do so, we let $\Phi_\mathrm{T}$ be a subspace of the Hilbert space $\mathcal{L}^2([0, T])$, together with its induced norm. For clarity of exposition, the inner product is denoted $\langle \prescript{}{1}{\varphi_T}, \prescript{}{2}{\varphi_T} \rangle$.

\subsection{Linear Independence}

One crucial requirement in Theorem \ref{theo: parameter estimation} that is often assumed, but not evaluated, is the linear independence of MFs. As an example, consider the polynomial TMFs
\begin{gather}\label{eq: extra example poly TMF1}
    \prescript{}{1}{\varphi}_{\mathrm{T}}(\tau)=\tau^2 (T-\tau)^2,\; \; \;
    \prescript{}{2}{\varphi}_{\mathrm{T}}(\tau)=\tau^2 (T-\tau)^3,\\
    \prescript{}{3}{\varphi}_{\mathrm{T}}(\tau)=\tau^3 (T-\tau)^2,\; \; \;
    \prescript{}{4}{\varphi}_{\mathrm{T}}(\tau)=\tau^3 (T-\tau)^3,\\
    \prescript{}{5}{\varphi}_{\mathrm{T}}(\tau)=\tau^4 (T-\tau)^4, \label{eq: extra example poly TMF5}
\end{gather}
and define the angle between the functions as
\begin{equation}
    \cos{\alpha_{ij}} := \frac{\langle \prescript{}{i}{\varphi}_{\mathrm{T}}, \prescript{}{j}{\varphi}_{\mathrm{T}} \rangle } {||\prescript{}{i}{\varphi}_{\mathrm{T}} || ~ ||\prescript{}{j}{\varphi}_{\mathrm{T}} ||}.
\end{equation}
By direct computation, one obtains, for example,
\begin{gather}
    \cos{\alpha_{12}}\approx0.9574, \; \; \; \cos{\alpha_{13}}\approx 0.9574 \\
    \cos{\alpha_{14}}\approx0.9924, \; \; \; \cos{\alpha_{15}}\approx 0.9774,
\end{gather}
indicating that, despite the functions not being collinear, they are \textit{almost} collinear. Similarly, computing the angle $\beta_{ij}$ between the derivative of the MFs gives, e.g.,
\begin{equation}
    \cos{\beta_{45}} = \frac{\langle \prescript{}{i}{\varphi}'_{\mathrm{T}}, \prescript{}{j}{\varphi}'_{\mathrm{T}} \rangle } {||\prescript{}{i}{\varphi}'_{\mathrm{T}} || ~ ||\prescript{}{j}{\varphi}'_{\mathrm{T}} ||}\approx 0.9806.
\end{equation}
Such selections lead to sharp peaks in the estimates due to poor matrix conditioning: see \cite{Co_Ungarala_Batch_Recursive_Parameter_Estimation_MFM, IonesiRJ19}; and \cite{Li_Pin_Kernel_Based_Simulatenous_Parameter_State_Estimation}. 
A simple way to address this issue is to make use of Theorem \ref{theo: build MF} and select MFs from different families, which also leads to ``less similar'' derivatives. To illustrate this, consider the $\prescript{}{i}{\psi}_{\mathrm{T}}$ TMFs 
\begin{gather}
    \prescript{}{1}{\psi}_{\mathrm{T}}(\tau)=\tau^3 (T-\tau)^4,\label{eq: tmf psi1}\\
    \prescript{}{2}{\psi}_{\mathrm{T}}(\tau)=\tanh^3(3\tau) \tanh^3(1.5(T-\tau)) ,\\
    \prescript{}{3}{\psi}_{\mathrm{T}}(\tau)=\left( 5-\sinh(4 \frac{\tau}{T})\right)~h\Bigl( \frac{2\tau}{T}-1\Bigr),\\
    \prescript{}{4}{\psi}_{\mathrm{T}}(\tau)=\tanh(3 \tau)~h\Bigl( \frac{2\tau}{T}-1\Bigr), \label{eq: tmf psi4}
\end{gather}
having $h(\tau)$ defined according to \eqref{eq: bump test function}, together with
\begin{equation}\label{eq: TMF candidate 5}
    \begin{aligned}
        \prescript{}{5}{\psi}_{\mathrm{T}}(\tau)=&e^{\tau}\sin^3\left(\frac{3\pi}{T} \tau\right)+e^{-\tau}\sin^4\left(\frac{4\pi}{T} \tau\right)\\
        &+(\ln (5 \tau+1))^{3} (\ln (9.1 \tau+1) - \ln(9.1 T+1))^{3}\\
        &+(\sech(4.5 \tau)-1)^{3} (\sech(-2.7(T-\tau))-1)^{3}\text{,}
    \end{aligned}
\end{equation}
which is at least of order $3$ due to Corollary \ref{corollary: increase and decrease order of MFs}.
By direct computation, it follows that
\begin{gather}\label{eq: angle 12 and 13 for psi TMFs}
    \cos{\alpha_{12}}\approx0.8121, \; \; \; \cos{\alpha_{13}}\approx 0.4196 \\
    \cos{\alpha_{14}}\approx0.9440, \; \; \; \cos{\alpha_{15}}\approx 0.016, \label{eq: angle 14 and 15 for psi TMFs}
\end{gather}
and the angle between the derivatives of $\prescript{}{4}{\psi}_{\mathrm{T}}$ and $\prescript{}{5}{\psi}_{\mathrm{T}}$ is $\cos{\beta_{45}} \approx -0.0913$. Thus, by simply selecting the TMFs from different families, it is possible to avoid ``almost collinear'' scenarios and, consequently, improve matrix conditioning: a strategy applicable to systems of any order thanks to Theorem \ref{theo: build MF}.

\subsection{Orthonormal Total Modulating Functions} \label{subsec: orthogonal TMFs}

Although selecting MFs from different families does improve the collinearity issue, a further refinement is possible in, e.g., $\alpha_{12}$ and $\alpha_{14}$ in \eqref{eq: angle 12 and 13 for psi TMFs}-\eqref{eq: angle 14 and 15 for psi TMFs}: an improvement that can be obtained by further exploiting the algebraic properties of TMFs to perform a Gram-Schmidt process.

To be precise, start with a TMF $\prescript{}{1}{\psi_\mathrm{T}}(\tau)$ of order $q_1\geq q$, and use Theorem \ref{theo: closedness MFs} to select the first TMF as\footnote{The normalization of MFs is common in the community. See, e.g., \cite{PreisigSplinesII} and \cite{LiuKPG14}.}
\begin{equation}
    \prescript{}{1}{\varphi_\mathrm{T}}(\tau)=\frac{\prescript{}{1}{\psi_\mathrm{T}}(\tau)}{||\prescript{}{1}{\psi_\mathrm{T}}||}\text{.}
\end{equation}
Next, take a second TMF $\prescript{}{2}{\psi_\mathrm{T}}(\tau)$ of order $q_2\geq q$ and, again, use Theorem \ref{theo: closedness MFs} and Corollary \ref{corollary: increase and decrease order of MFs} to construct the TMF $\prescript{}{2}{\varphi_\mathrm{T}}(\tau)$: first, note that $\prescript{}{2}{\bar{\varphi}_\mathrm{T}}(\tau)$ is given by
\begin{equation}
    \prescript{}{2}{\bar{\varphi}_\mathrm{T}}(\tau)=\prescript{}{2}{\psi_\mathrm{T}}(\tau)-\langle \prescript{}{1}{\varphi_\mathrm{T}}, \prescript{}{2}{\psi_\mathrm{T}} \rangle \prescript{}{1}{\varphi_\mathrm{T}}(\tau)
\end{equation}
and is also a TMF. Then, define $\prescript{}{2}{\varphi_\mathrm{T}}(\tau)$ as $\prescript{}{2}{\varphi_\mathrm{T}}(\tau)=\frac{\prescript{}{2}{\bar{\varphi}_\mathrm{T}}(\tau)}{||\prescript{}{2}{\bar{\varphi}_\mathrm{T}}||}$. Simply continue the process to obtain $\prescript{}{3}{\varphi_\mathrm{T}}(\tau)$ from the TMF $\prescript{}{3}{\psi_\mathrm{T}}(\tau)$ of order $q_3\geq q$ as
\begin{equation}
    \prescript{}{3}{\bar{\varphi}_T}=\prescript{}{3}{\psi_T}-\langle\prescript{}{1}{\varphi_\mathrm{T}}, \prescript{}{3}{\psi_\mathrm{T}} \rangle \prescript{}{1}{\varphi_\mathrm{T}} - \langle\prescript{}{2}{\varphi_\mathrm{T}}, \prescript{}{3}{\psi_\mathrm{T}} \rangle \prescript{}{2}{\varphi_\mathrm{T}}
\end{equation}
along with $\prescript{}{3}{\varphi_\mathrm{T}}(\tau)=\frac{\prescript{}{3}{\bar{\varphi}_\mathrm{T}}(\tau)}{||\prescript{}{3}{\bar{\varphi}_\mathrm{T}}||}$,
repeating the process for as many TMFs as desired. By following this process, all obtained TMFs $\prescript{}{i}{\varphi_\mathrm{T}}(\tau)$ are orthogonal to each other and have unit norm. Note also that, for TMFs $\prescript{}{i}{\psi_\mathrm{T}}(\tau)$ of order $q_i\geq q$, all MFs generated through this process have order larger or equal to $q$, according to Corollary \ref{corollary: increase and decrease order of MFs}. \\
For example, using the candidates \eqref{eq: tmf psi1}-\eqref{eq: TMF candidate 5}, it trivially follows that $\cos \alpha_{ij}=0$ for $i\neq j$ ensuring that more information is added with each new TMF due to the orthogonality of the functions, while also avoiding simultaneous orthogonality to the measured signals. Moreover, due to Corollary \ref{corollary: increase and decrease order of MFs}, all resulting TMFs are of order $q_i \geq 3$. However, it now follows that $\cos \beta_{45}\approx 0.3179$.
\section{Application Example}\label{sec: application example}

\begin{figure}[t]
    \centering
    \includegraphics[width=0.55\linewidth]{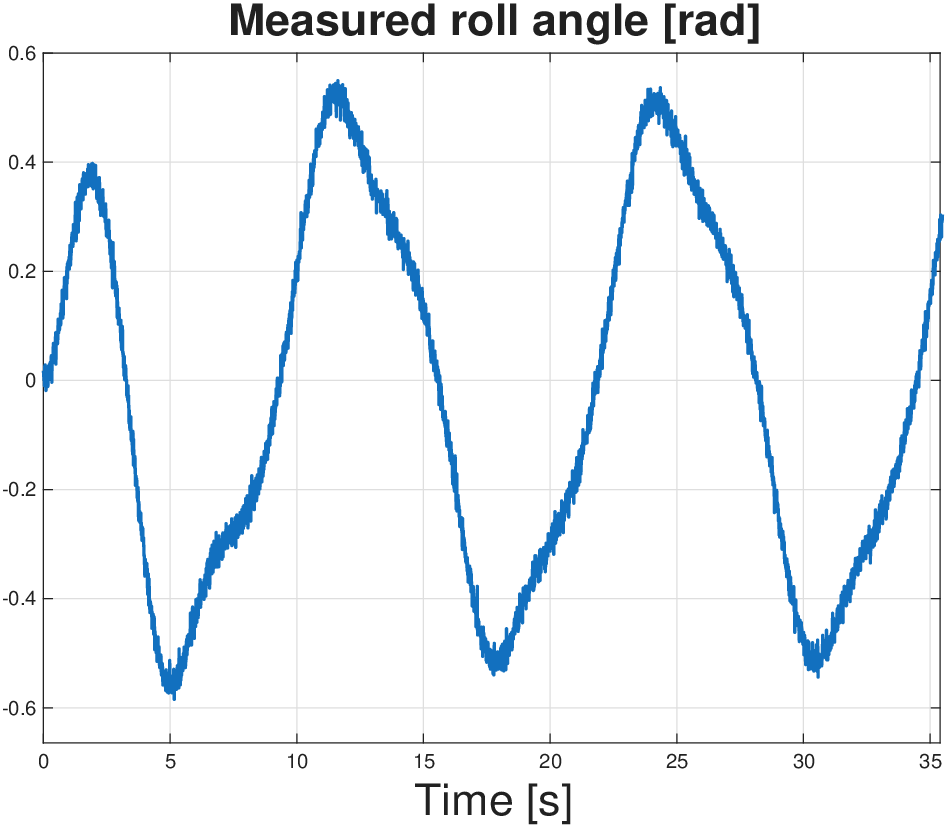}
    \caption{Measured roll angle of system \eqref{eq: boat dynamics}.}
    \label{fig: measured roll angle}
\end{figure}

\begin{figure*}[ht]
    \centering
    \includegraphics[width=0.24\linewidth]{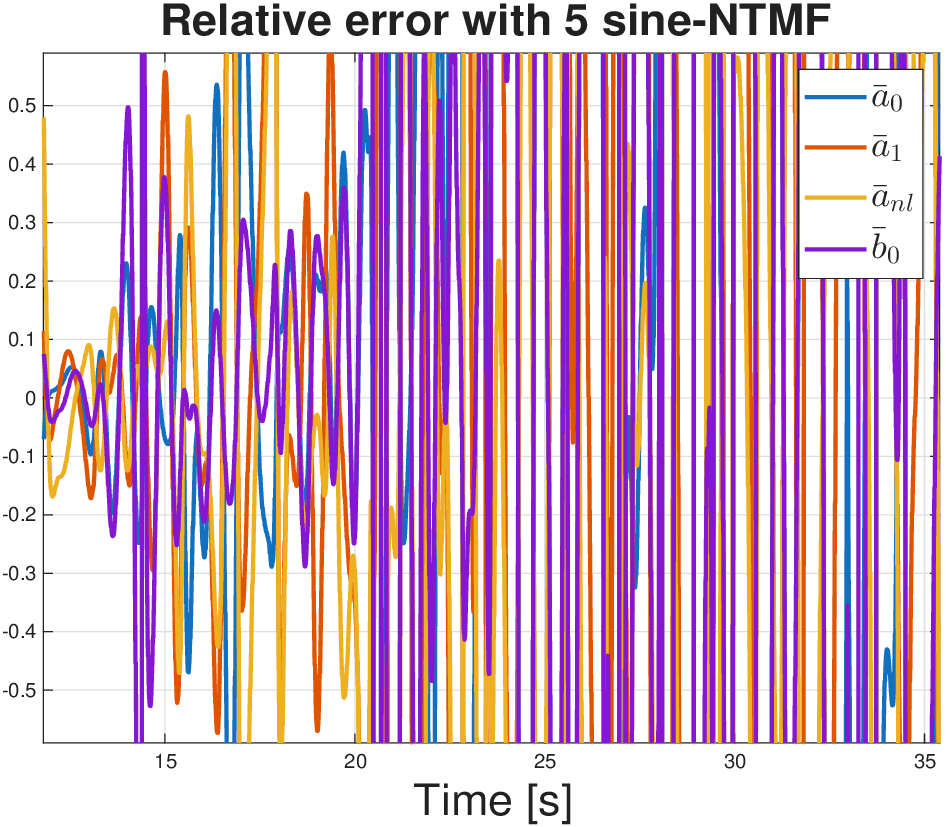}
    \includegraphics[width=0.24\linewidth]{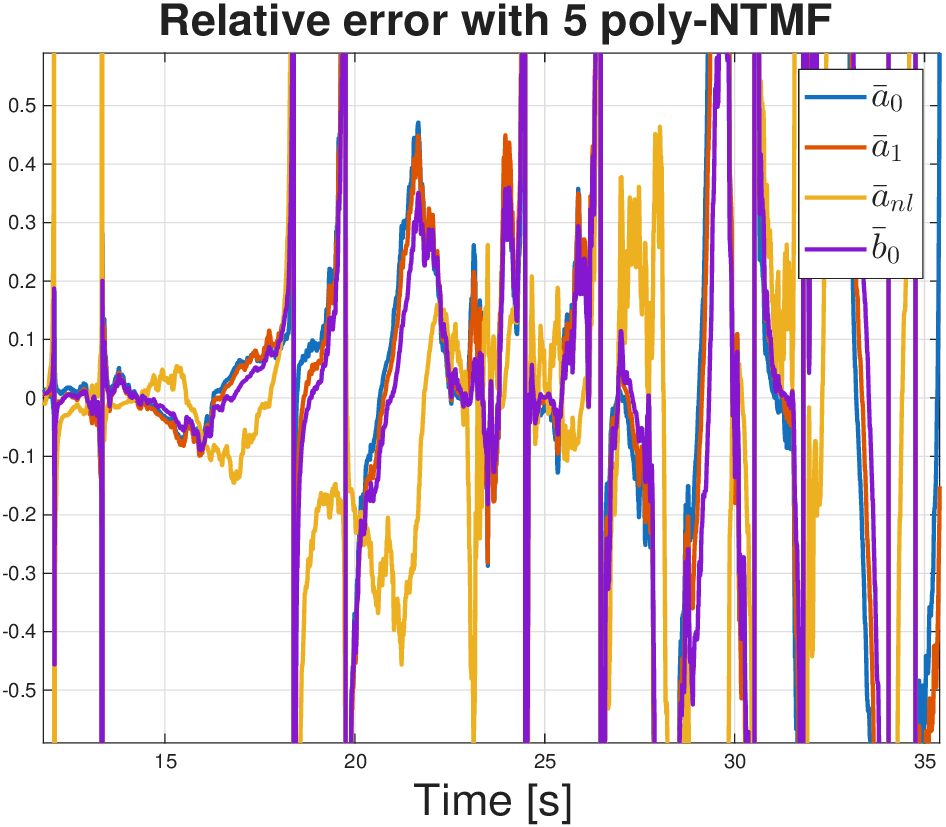}
    \includegraphics[width=0.24\linewidth]{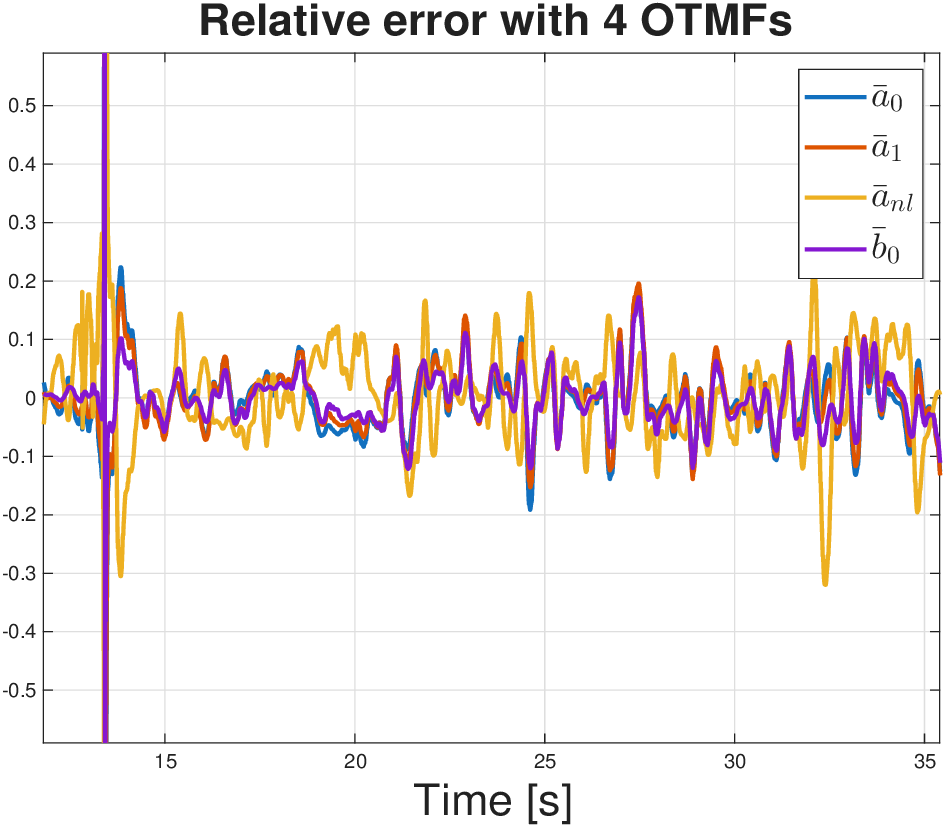}
    \includegraphics[width=0.24\linewidth]{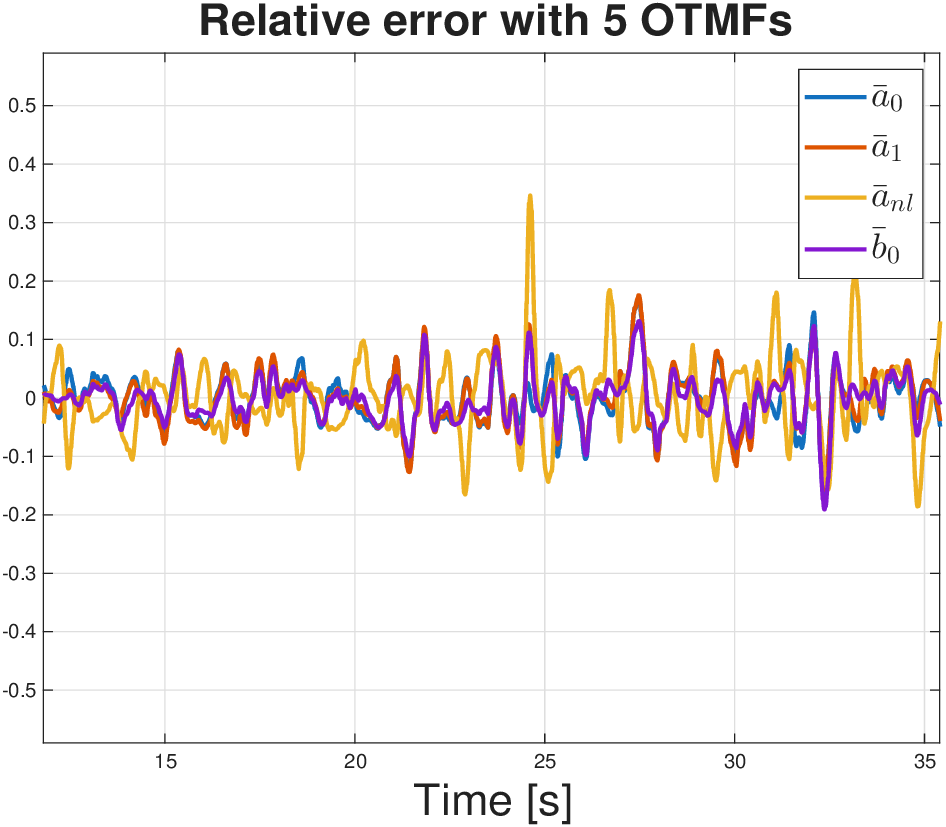}
    \caption{Relative estimation error of the parameter estimation problem of system \eqref{eq: boat dynamics} for $T=11.8s$ with: 5 sine-type TMFs, i.e. \eqref{eq: sine TMF family}, of order 2-5, which were normalized; the 5 polynomial-type TMFs given in \eqref{eq: extra example poly TMF1}-\eqref{eq: extra example poly TMF5}, also being normalized; the 4 orthonormal TMFs obtained using candidates \eqref{eq: tmf psi1}-\eqref{eq: tmf psi4}, and the 5 orthonormal TMFs using candidates \eqref{eq: tmf psi1}-\eqref{eq: TMF candidate 5}.}
    \label{fig: parameter estimation results}
\end{figure*}

\begin{figure*}[ht]
    \centering
    \includegraphics[width=0.24\linewidth]{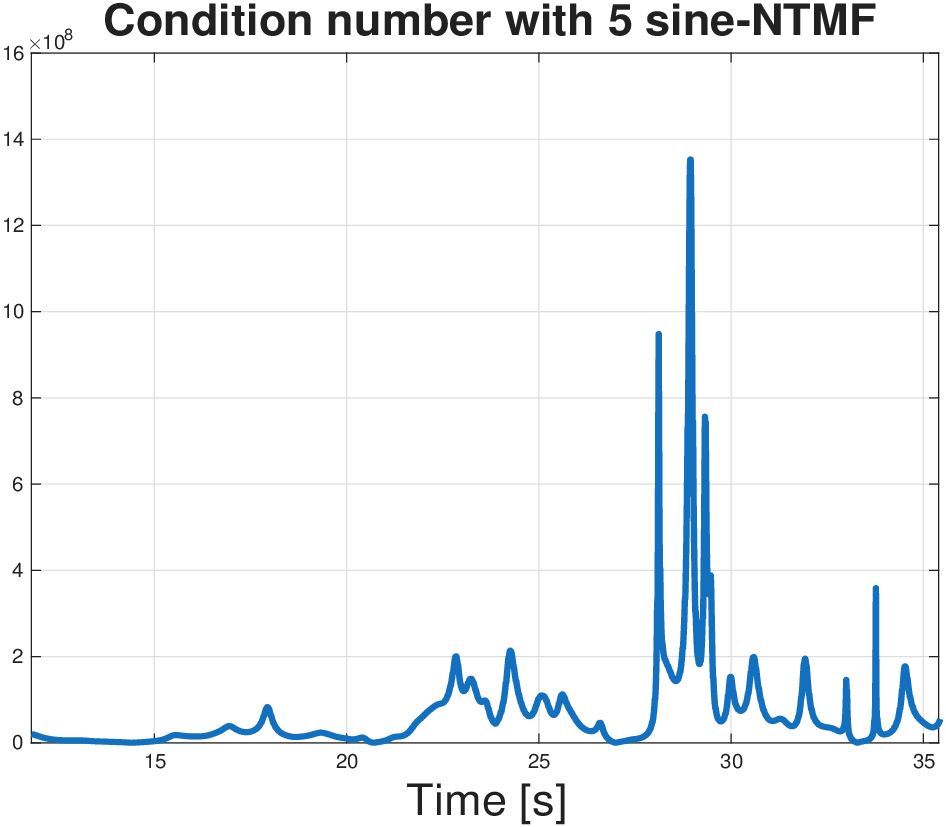}
    \includegraphics[width=0.24\linewidth]{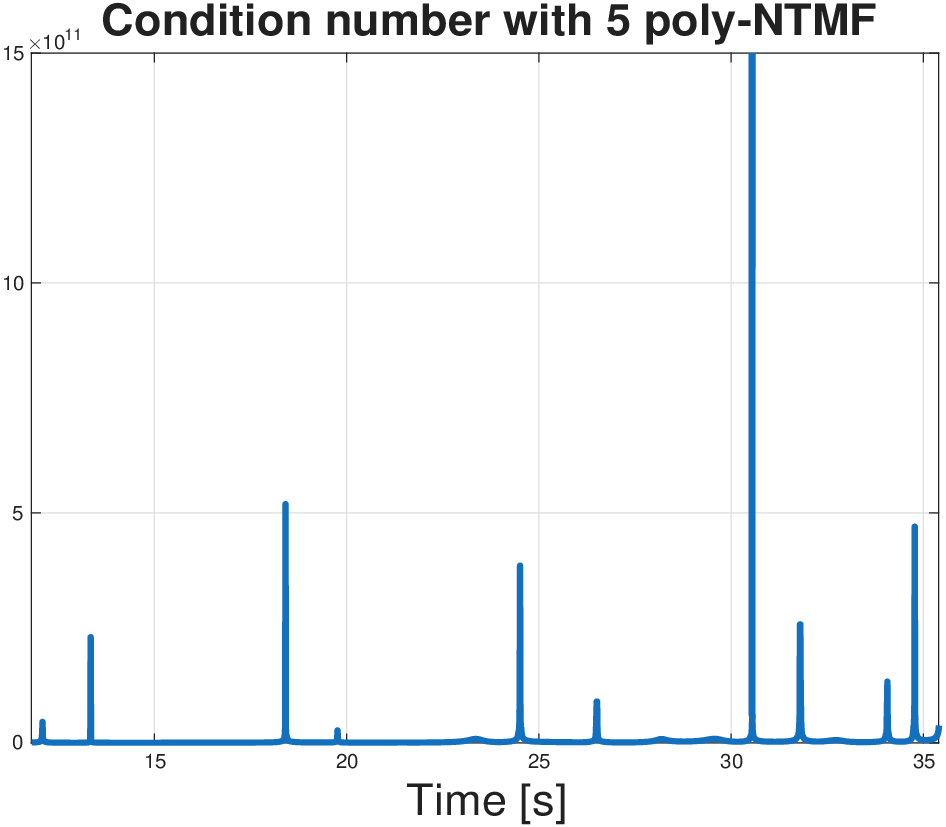}
    \includegraphics[width=0.24\linewidth]{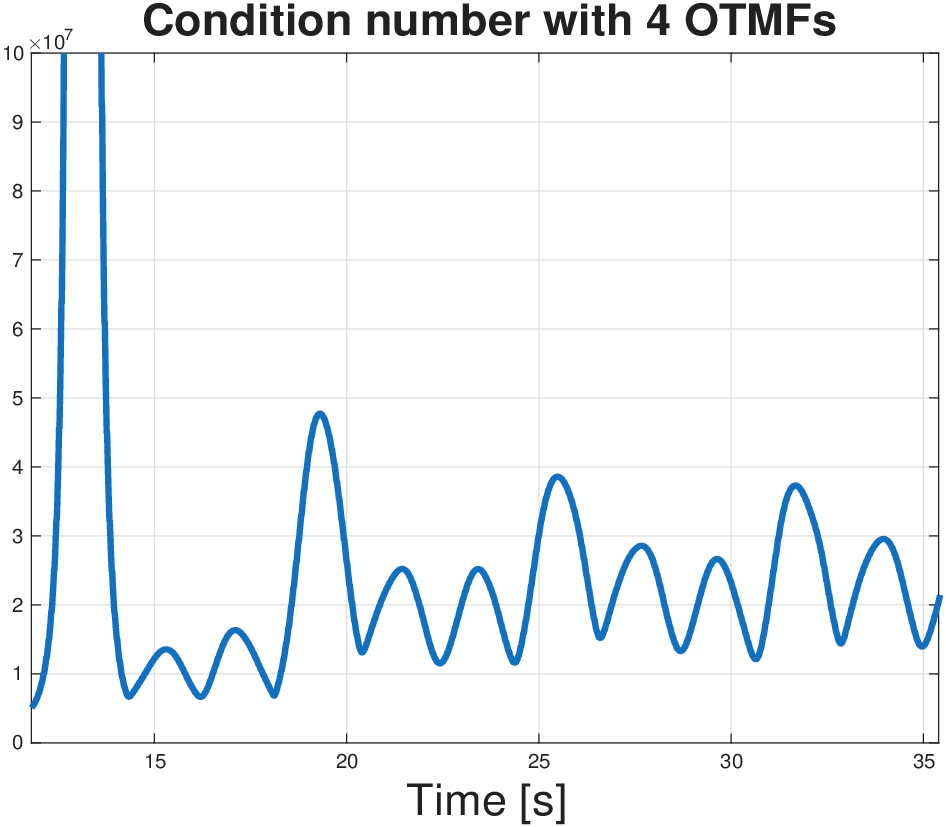}
    \includegraphics[width=0.24\linewidth]{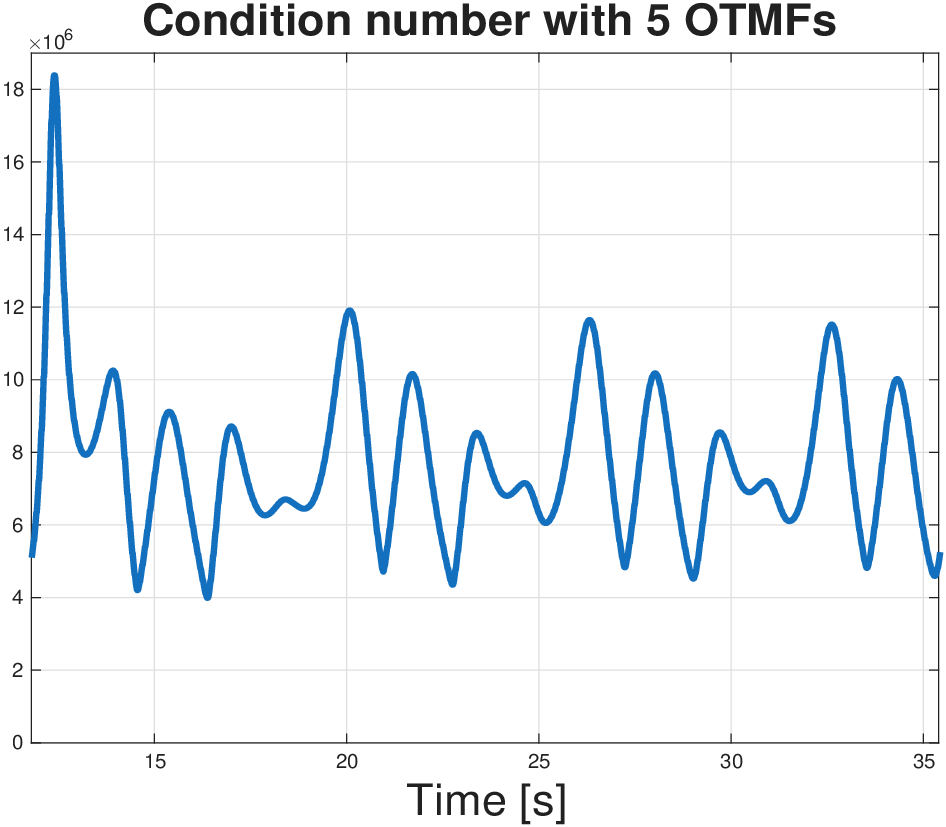}
    \caption{Condition number of the parameter estimation problem of system \eqref{eq: boat dynamics} for $T=11.8s$ with: 5 sine-type TMFs, i.e. \eqref{eq: sine TMF family}, of order 2-5, which were normalized; the 5 polynomial-type TMFs given in \eqref{eq: extra example poly TMF1}-\eqref{eq: extra example poly TMF5}, also being normalized; the 4 orthonormal TMFs obtained using candidates \eqref{eq: tmf psi1}-\eqref{eq: tmf psi4}, and the 5 orthonormal TMFs using candidates \eqref{eq: tmf psi1}-\eqref{eq: TMF candidate 5}.}
    \label{fig: condition number with 5 TMFs}
\end{figure*}

As a practical application of these concepts, consider the parameter estimation problem of a boat's roll dynamics at low angular speeds, described by \citep{Ramezani_Roll_Stabilization_Canting_Keel}
\begin{equation}\label{eq: boat dynamics}
    \ddot{\phi}+a_1 \dot{\phi}+a_0 \phi  +a_{nl} \phi^3= b_0 \mathrm{u} \text{,} 
\end{equation}
where the input torque $u(t) \in \mathbb{R}$ and the boat roll angle $\phi(t) \in \mathbb{R}$ are assumed to be measured. Note that, despite this system being nonlinear, the nonlinearity only depends on known variables and, therefore, can still be handled using Theorem \ref{theo: parameter estimation} by considering it as another input (see, e.g., \cite{Acc_Hyperbolic_MF}).

The coefficients $a_0$, $a_1$, $a_{nl}$, $b_0$ are related to the density of the fluid $\rho$, the acceleration of gravity $g$, the volume displacement of the ship $\nabla$, the mean metacentric height $\overline{GM}$, the linear drag coefficient $K_{\dot{\phi}}$, the nonlinear restoring coefficient $K_{\phi^3}$, the roll inertia $I_x$ and the added mass coefficient $\Delta I_x$ through the relations
\begin{equation}
    \begin{array}{lr}
         a_0=\dfrac{\rho g \nabla \overline{GM}}{I_x+\Delta I_x},\; \; \; \; \; \; & a_1=\dfrac{K_{\dot{\phi}}}{I_x + \Delta I_x}, \\
         a_{nl}=\dfrac{K_{\phi^3}}{I_x + \Delta I_x}, \; \; \; \; \; \; & b_0=\dfrac{1}{I_x+\Delta I_x}\text{.}
    \end{array}
\end{equation}
Note also that, since $b_0^{-1}$ gives the total roll inertia, coefficients such as $K_{\dot{\phi}}$ and $K_{\phi^3}$ can be individually obtained using the estimates $\hat{b}_0$, $\hat{a}_1$, and $\hat{a}_{nl}$. Similarly, $\overline{GM}$ can also be estimated using $\hat{b}_0$ and $\hat{a}_0$ since $\rho$ and $g$ are known constants. Thus, it is possible to, e.g., estimate the impact of moving cargo on the boat's moment of inertia.

The true model parameters were then selected as $a_0=1.33$, $a_1=0.64$, $a_{nl}=2.43$, and $b_0=6.4 \cdot 10^{-6}$, being obtained from the canting keel parameters given in \cite{Ramezani_Roll_Stabilization_Canting_Keel}, while the input is given as $\mathrm{u}(t)= 115625\cos(0.5 t)$ and zero-mean Gaussian noise with standard deviation $\sigma=0.015$ is added to the output. Such noise levels were selected based on existing literature for attitude estimation, such as \cite{Mahony_Nonlinear_Complementary_Filter}. Lastly, the initial conditions were selected as $\phi(0)=\dot{\phi}(0)=0$.

To estimate the 4 parameters, 5 sine-type TMFs, i.e. \eqref{eq: sine TMF family} with $q=\{2, 3, 4, 5, 6\}$, and the 5 polynomial-type TMFs in \eqref{eq: extra example poly TMF1}-\eqref{eq: extra example poly TMF5}, all of which were normalized and denoted as normalized TMFs (NTMF), are used as a baseline for comparison. Next, a set of 4 and 5 orthonormal TMFs (OTMF), obtained respectively from \eqref{eq: tmf psi1}-\eqref{eq: tmf psi4} and \eqref{eq: tmf psi1}-\eqref{eq: TMF candidate 5}, were constructed using the Gram-Schmidt procedure, and numerical simulations were run using Theorem \ref{theo: parameter estimation}, having the relative estimation error shown in Figure \ref{fig: parameter estimation results}.

\begin{figure*}[ht]
    \centering
    \includegraphics[width=0.24\linewidth]{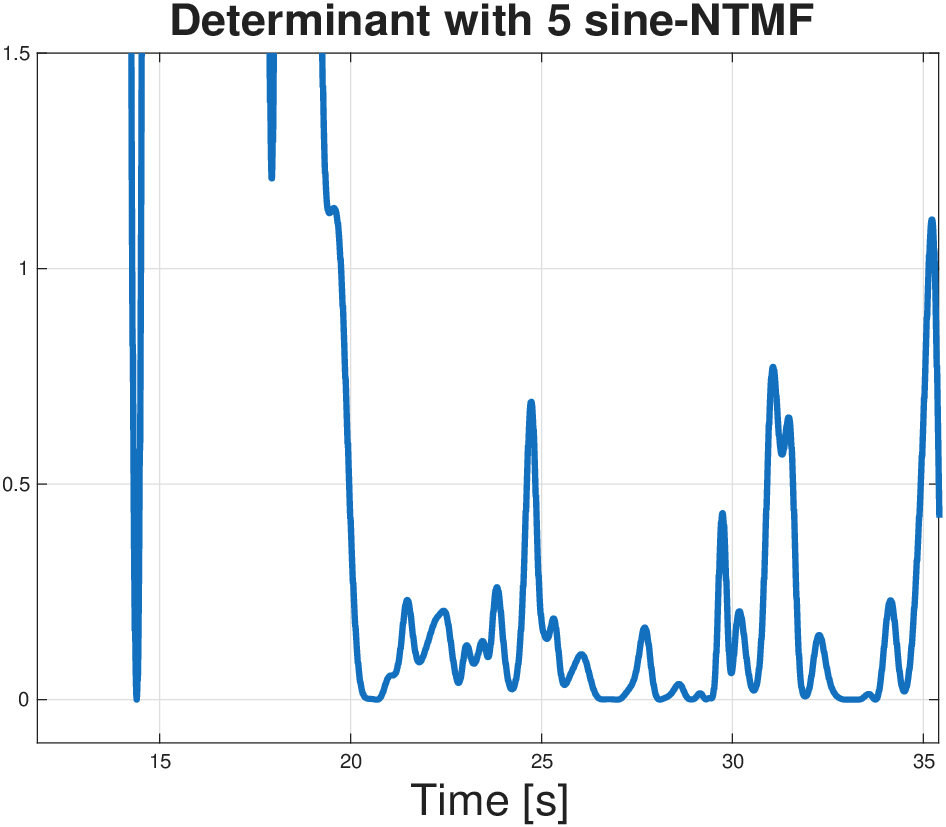}
    \includegraphics[width=0.24\linewidth]{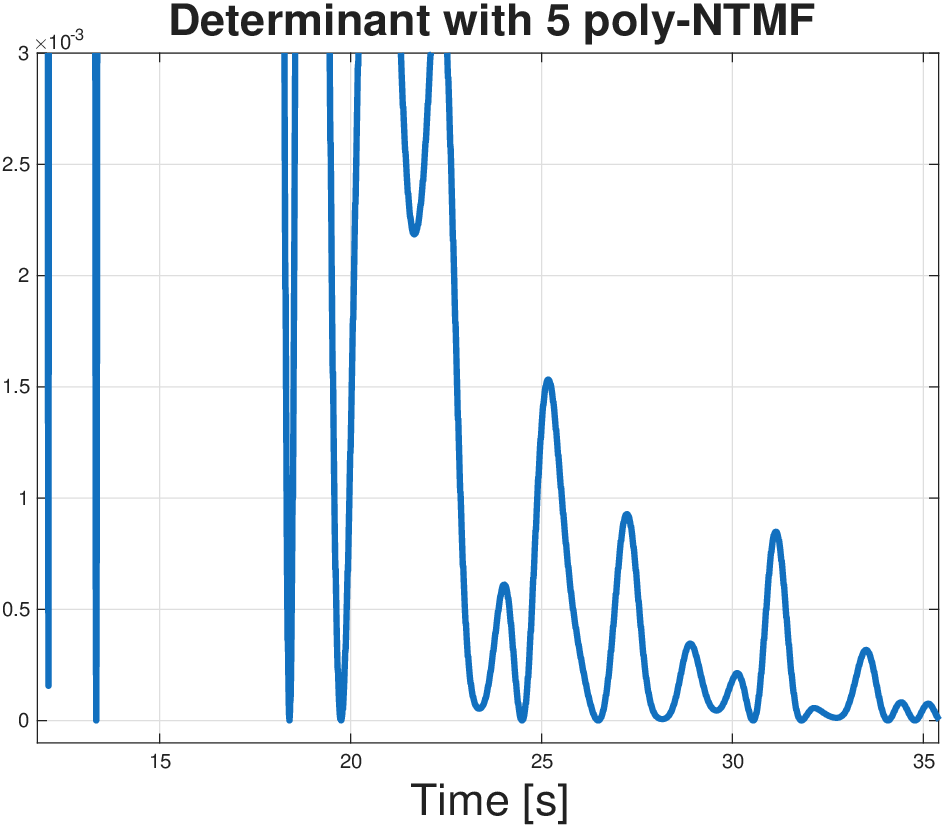}
    \includegraphics[width=0.24\linewidth]{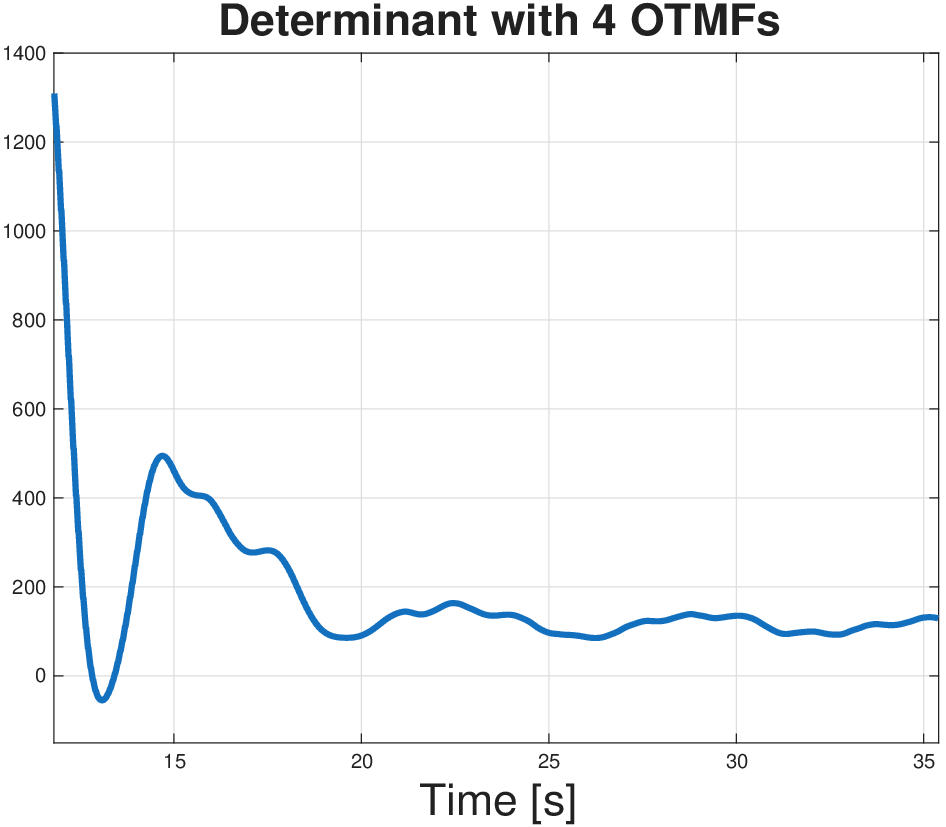}
    \includegraphics[width=0.24\linewidth]{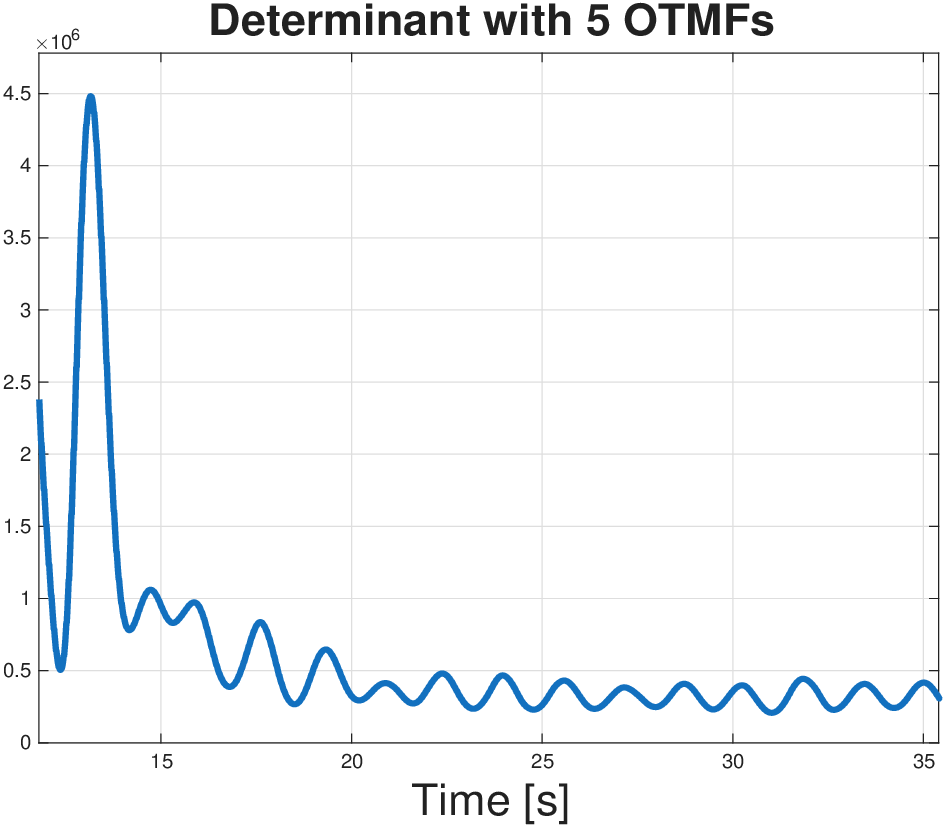}
    \caption{Determinant of the parameter estimation problem of system \eqref{eq: boat dynamics} for $T=11.8s$ with: 5 sine-type TMFs, i.e. \eqref{eq: sine TMF family}, of order 2-5, which were normalized; the 5 polynomial-type TMFs given in \eqref{eq: extra example poly TMF1}-\eqref{eq: extra example poly TMF5}, also being normalized; the 4 orthonormal TMFs obtained using candidates \eqref{eq: tmf psi1}-\eqref{eq: tmf psi4}, and the 5 orthonormal TMFs using candidates \eqref{eq: tmf psi1}-\eqref{eq: TMF candidate 5}.}
    \label{fig: determinant with 5 TMFs}
\end{figure*}

\begin{figure*}[ht]
    \centering
    \includegraphics[width=0.24\linewidth]{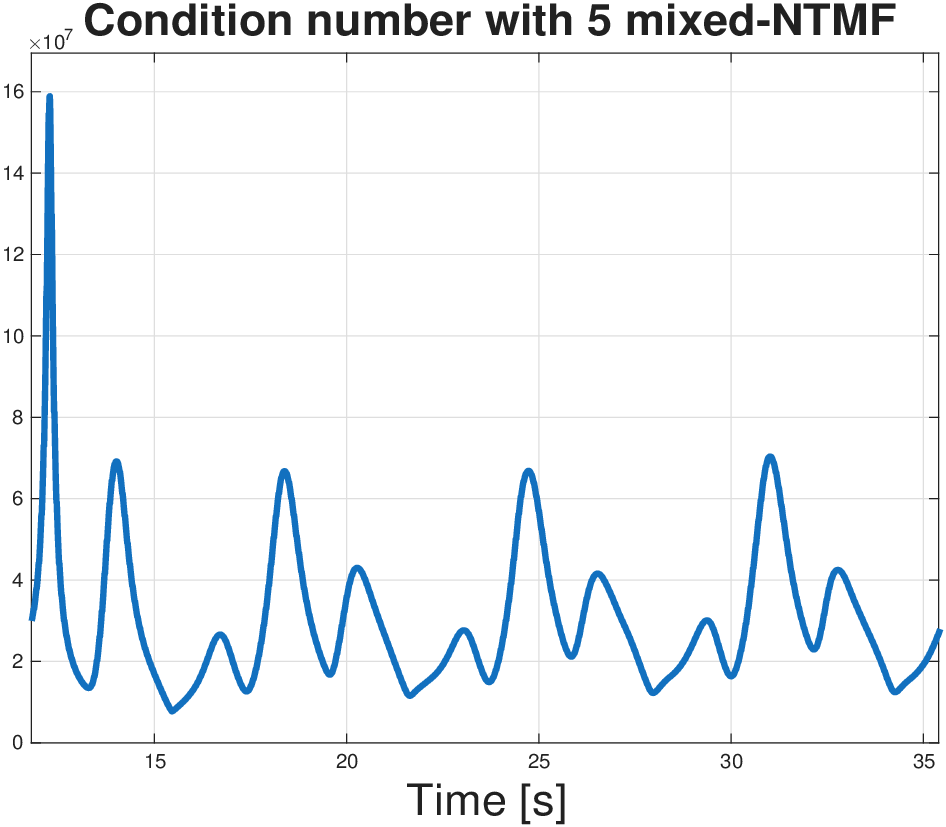}
    \includegraphics[width=0.24\linewidth]{Figures/cond_W_5_orthonormal_TMFs.eps}
    \includegraphics[width=0.24\linewidth]{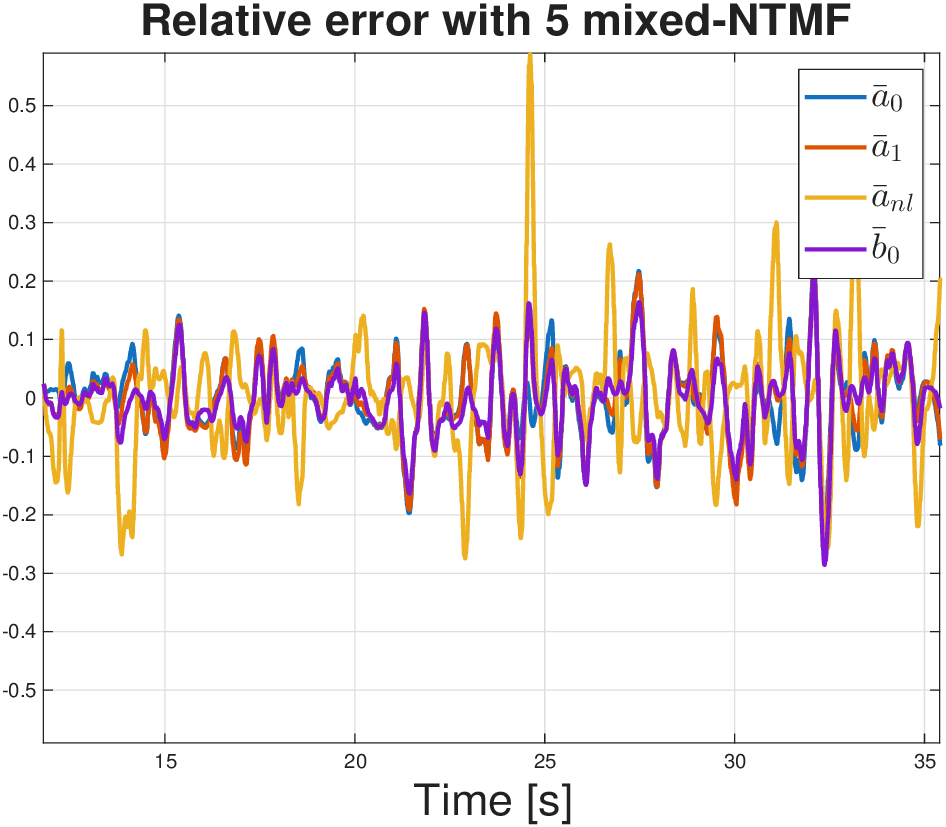}
    \includegraphics[width=0.24\linewidth]{Figures/relative_err_param_est_5_TMFs.eps}
    \caption{Comparison between condition number and relative error of normalized TMFs of different families, namely \eqref{eq: tmf psi1}-\eqref{eq: TMF candidate 5}, and the orthonormal TMFs generated by the same candidates, both with $T=11.8s$.}
    \label{fig: comparison mixed TMFs and OTMFs}
\end{figure*}

The TMF selection impact is obvious in Figure \ref{fig: parameter estimation results}: even when only 4 OTMFs are utilized, the relative error is significantly smaller and more well-behaved than the results obtained with 5 sine NTMFs and 5 polynomial NTMFs. To be precise, the norm of the RMS relative estimation error, henceforth ``norm RMS error'', using the 5 sine NTMFs was $33.4814$ and for the 5 polynomial NTMFs was $14.9695$. In contrast, the norm RMS error obtained with the 4 OTMFs was $0.2956$, an improvement of over $98\%$ relative to the baseline, and $0.0998$ with 5 OTMFs, an improvement of over $99\%$ relative to the baseline, without increasing online computational costs.

This improvement directly follows from lower condition numbers of $\mathbf{W}$, which showed orders of magnitude of improvement, as seen in Figure \ref{fig: condition number with 5 TMFs}. It is clearly shown that, in general, using OTMFs obtained from different families decreases the estimation error and can effectively avoid the well-known singularity issue of the method seen in many previous contributions, such as  \cite{Co_Ungarala_Batch_Recursive_Parameter_Estimation_MFM}, \cite{IonesiRJ19}, and \cite{Li_Pin_Kernel_Based_Simulatenous_Parameter_State_Estimation}. Another important indicator of the solvability of \eqref{eq: parameter estimation} is the determinant of $\mathbf{W}$, or of the normalized problem $\mathbf{W}^\top \mathbf{W}$ for non-square matrices, which is shown in Figure \ref{fig: determinant with 5 TMFs}. It is again evident that the use of OTMFs improved the matrix conditioning, since the determinant never reaches zero for 5 OTMFs and only has a brief period of singularity with 4 OTMFs: a significant improvement over the baseline NTMFs.

Lastly, it is also important to evaluate the improvement of the orthonormality of TMFs when compared to solely using TMFs of different families. For this reason, the TMF candidates $\prescript{}{i}{\psi}_{\mathrm{T}}$ given in \eqref{eq: tmf psi1}-\eqref{eq: TMF candidate 5} were normalized and the results of the estimation process obtained with them were compared to the ones obtained with 5 OTMFs, being shown in Figure \ref{fig: comparison mixed TMFs and OTMFs}. As discussed in Section \ref{sec: matrix conditioning}, using different families of MFs significantly improves the matrix conditioning when compared to the baseline, but the condition number is still one order of magnitude larger than the 5 OTMF case, although comparable to the 4 OTMF case. In particular, the norm RMS error obtained with 5 NTMFs of different families was $0.1485$, which is smaller than the one obtained with 4 OTMFs but larger than the one obtained with 5 OTMFs, indicating the benefit of using OTMFs instead of just normalized TMFs from different families.
\section{Concluding Remarks}
\label{sec: conclusion}

In this paper, the algebraic properties of modulating functions were formalized and its consequences were discussed. In particular, simple algorithms to modify and construct new modulating functions from existing ones were proposed, along with an algorithm to generate new families of modulating functions from any sufficiently smooth function. Moreover, by exploiting the vector space properties of total modulating functions, it was shown that selecting modulating functions from different families significantly improves matrix conditioning, and using the newly proposed orthonormal total modulating functions further improve the numerical stability of the method. Lastly, the proposed concepts were illustrated for a boat's roll dynamics.

While the results proposed in this paper do not solve the general problem of selecting modulating functions, they hopefully provide a step towards this goal. Further research on this topic could involve extending the presented analysis to the general problem of modulating function selection, or comparing the presented results with auxiliary systems in hardware implementations.

\printbibliography

\end{document}